\documentclass[copyright,creativecommons]{eptcs}
\providecommand{\event}{FICS 2024} % Name of the event you are submitting to

\usepackage{iftex}
\usepackage{amsmath}
\usepackage{amssymb}
\usepackage{amsthm}
\usepackage{xspace}
\usepackage{mathtools}
\usepackage{stmaryrd}
\usepackage{algorithm}
\usepackage{algpseudocode}
\usepackage{xcolor}

\ifpdf
  \usepackage{underscore}         % Only needed if you use pdflatex.
  \usepackage[T1]{fontenc}        % Recommended with pdflatex
\else
  \usepackage{breakurl}           % Not needed if you use pdflatex only.
\fi

\theoremstyle{plain}
\newtheorem{theorem}{Theorem}
\newtheorem{lemma}[theorem]{Lemma}

\theoremstyle{definition}
\newtheorem{definition}[theorem]{Definition}

\theoremstyle{remark}
\newtheorem{remark}[theorem]{Remark}

\renewcommand{\epsilon}{\varepsilon}

\newcommand{\twr}[2]{2^{#1}_{#2}}

\newcommand{\blank}{\square}
\newcommand{\qin}{q_{\text{init}}}
\newcommand{\qyes}{q_{\text{acc}}}
\newcommand{\qno}{q_{\text{rej}}}

\newcommand{\grtype}{\bullet}
\newcommand{\ord}[1]{\mathit{ord}(#1)}
\newcommand{\vars}{\mathcal{V}}
\newcommand{\hop}{\ensuremath{\text{HO}\text{+}{\text{PFP}}\xspace}}
\newcommand{\HOP}[1]{\ensuremath{\textsf{HO}^{#1}\text{+}\textsf{PFP}}\xspace}
\newcommand{\pfp}{\text{PFP}}

\newcommand{\Transsys}{\ensuremath{{T}}}
\newcommand{\States}{\ensuremath{{S}}}
\newcommand{\Prop}{\ensuremath{\mathbf{P}}}
\newcommand{\Act}{\ensuremath{\mathbf{A}}}
\newcommand{\Actions}{\ensuremath{A}}
\newcommand{\Label}{\ensuremath{\ell}}
\newcommand{\state}{\ensuremath{s}}
\newcommand{\statealt}{\ensuremath{t}}

\newcommand{\HOPFP}{\hopfp{}}
\newcommand{\hopfp}[1]{\ensuremath{\textsf{HO}^{#1}\text{+}\textsf{PFP}}\xspace}
\newcommand{\pfpx}[1]{\ensuremath{\text{PFP}\ #1}}

\newcommand{\sem}[3]{\ensuremath{\llbracket #1 \rrbracket^{#2}_{#3}}}
\newcommand{\semm}[2]{\sem{#1}{#2}{}}

\newcommand{\ExpTime}{\textsf{EXPTIME}\xspace}
\newcommand{\ExpSpace}{\textsf{EXPSPACE}\xspace}
\newcommand{\PTime}{\textsf{P}\xspace}
\newcommand{\PSpace}{\textsf{PSPACE}\xspace}

\def\newarrow#1{\mathop{{\hbox{\setbox0=\hbox{$\scriptstyle{#1\quad}$}{$%
\mathrel{\mathop{\setbox1=\hbox to
\wd0{\rightarrowfill}\ht1=3pt\dp1=-2pt\box1}\limits^{#1}}%
$}}}}}

\def\newarrowi#1#2{\mathop{{\hbox{\setbox0=\hbox{$\scriptstyle{#1\quad}$}{$%
\mathrel{\mathop{\setbox1=\hbox to
\wd0{\rightarrowfill}\ht1=3pt\dp1=-2pt\box1}\limits^{#1}}_{#2}%
$}}}}}

\newcommand{\Transition}[3]{\ensuremath{#1 \newarrow{#2} #3}}

\newcommand{\mytrue}{\ensuremath{\mathtt{t\!t}}}
\newcommand{\myfalse}{\ensuremath{\mathtt{f\!f}}}

\algblockdefx{Cases}{EndCases}%
   [1]{\textbf{case} #1 \textbf{of}}%
   {\textbf{end case}}
\algcblockx[Cases]{Cases}{Case}{EndCases}%
   [1]{#1:\enspace}%
   {\textbf{end case}}

\title{Characterizing the Exponential-Space Hierarchy Via Partial Fixpoints}
\author{Florian Bruse \qquad \qquad David Kronenberger \qquad\qquad Martin Lange
\institute{University of Kassel, Germany}
\email{\{florian.bruse|martin.lange\}@uni-kassel.de}
}
\def\titlerunning{Characterizing the Exponential-Space Hierarchy Via Partial Fixpoints}
\def\authorrunning{F.\ Bruse, D.\ Kronenberger \& M.\ Lange}
\begin{document}
\maketitle

\begin{abstract}
The characterization of PSPACE-queries over ordered structures as exactly those expressible
in first-order logic with partial fixpoints (Vardi'82) is one of the classical
results in the field of descriptive complexity. In this paper, we extend this result to
characterizations of $k$-EXPSPACE-queries for arbitrary $k$, characterizing them as exactly
those expressible in order-$k+1$-higher-order logic with partial fixpoints. For $k > 1$, the 
restriction to ordered structures is no longer necessary due to the high expressive power of 
higher-order logic.
\end{abstract}

\section{Introduction}
\label{sec:intro}
% !TEX root =  main.tex

Computational complexity studies the difficulty of computation problems with regards to the consumption
of computational resources, most prominently time and space. Descriptive complexity, as a subdomain of
both computational complexity and formal logic, has taken this study to a more abstract level by 
characterizing classes of problems, i.e., complexity classes, through logical definability. This 
achieves the characterization of the difficulty, resp.\ complexity of problems without resorting to 
measuring the use of computational resources, as this ultimately depends on the choice of an underlying
model of computation like a Turing machine for instance. Descriptive complexity thus manages to characterize
the difficulty of problems through the structure of the problem alone, regardless of an underlying
model of computation. One can argue, though, that the resources used to measure complexity are logical 
operators that give the underlying logics their expressiveness, like predicate or fixpoint quantifiers.

Descriptive complexity started off with Fagin's seminal result \cite{Fagin74} showing that the well-known 
complexity class NP coincides with $\exists$SO, the set of problems definable in existential second-order 
logic. Stockmeyer extended this to a characterization of problems between NP and PSPACE by means of
second-order logic (SO), known as the polynomial hierarchy (PH) \cite{Stockmeyer:1976:PH}.

An interesting -- and still open -- question asks for a logical characterization of the complexity class 
P. This is believed to open ways to tackle the famous P=NP question. One of the major obstacles here is
the lack of a total order on the elements of a structure forming an instance of some computational problem,
like a graph for instance. When processing graphs with a computational model like a Turing machine, it
can be assumed to be totally ordered due to the way that it needs to be represented as an input. For
logical formulas, operating directly on structures and not on string representations thereof, this is
not the case. On the other hand, a total order helps immensely; it enables iteration over all elements
of the structure. Moreover, a logical characterization of the complexity class P is known when inputs
to its problems are assumed to be explicitly ordered. This is known as the Immerman-Vardi Theorem
\cite{DBLP:journals/iandc/Immerman86,DBLP:conf/stoc/Vardi82}, 
stating that the complexity class P on ordered structures is captured by the extension of first-order
logic with least fixpoint quantifiers (FO+LFP). 
 
Fixpoint quantifiers turned out to be a useful tool in descriptive complexity. Immerman lifted the
Immerman-Vardi Theorem to a characterization of the complexity class EXPTIME by second-order logic
with least fixpoint quantifiers (SO+LFP) \cite{Imm:lanccc}. Note that the fixpoint quantifiers in 
SO+LFP are not the same as the ones in FO+LFP. The LFP in FO+LFP refers to least fixpoints of 
first-order functions mapping tuples of elements to tuples of elements. This can be expressed in SO,
i.e.\ FO+LPF $\subseteq$ SO. The LFP in SO+LFP refers to fixpoints of second-order functions, mapping
predicates to predicates. This naturally gives rise to the question after characterizations of classes
in the exponential-time hierarchy by means of higher-order logic with fixpoints. Indeed, Freire
and Martins \cite{FREIRE201171} showed that for any $k \ge 2$, the class $k$-EXPTIME of problems 
solvable in $k$-fold exponential time is captured by HO$^{k+1}$+LFP, i.e.\ higher-order formulas of 
order at most $k+1$ with corresponding least fixpoint quantifiers.

Given the rather complete picture for time complexity, it is natural to ask whether space complexity
is also open to logical characterizations in the same fashion.
Another celebrated result in descriptive complexity, made use of in e.g., the Abiteboul-Vianu Theorem \cite{Ab-Vi}, 
is due to Vardi \cite{DBLP:conf/stoc/Vardi82} (not to be confused with the Immerman-Vardi Theorem, from 
\cite{DBLP:journals/iandc/Immerman86} and also \cite{DBLP:conf/stoc/Vardi82}). 
It states that the class PSPACE on ordered structures 
is captured by FO+PFP, i.e., the extension of first-order logic by \emph{partial} fixpoints. 

In this paper we extend the descriptive complexity of classes in the exponential space hierarchy
with the Vardi's result at the basis, just like Freire and Martins have done for the
time hierarchy with the Immerman-Vardi Theorem at its basis. We show that, for any $k \ge 1$,
the complexity class $k$-EXPSPACE of problems solvable using at most $k$-fold exponential space,
is captured by the logic HO$^{k+1}$+PFP of formulas of order at most $k+1$ with partial fixpoint
quantifiers.

\section{Preliminaries}
\label{sec:prelim}
% !TEX root =  main.tex

Let $n, k \in \mathbb{N}$. We write $\twr{n}{k}$ for the following: $\twr{n}{k} = n$ if $k = 0$, and $\twr{n}{k+1} = 2^{\twr{n}{k}}$.

\subsection{Space-bounded Turing Machines} 
A deterministic Turing machine (DTM) is a tuple $\mathcal{M} = (Q, \Sigma, \Gamma, \blank, \delta, \qin, \qyes, \qno)$ where $Q$ is a finite set of \emph{states},
$\Sigma$ is a finite, nonempty \emph{input alphabet}, $\Gamma \supseteq \Sigma$ is a finite, nonempty \emph{tape alphabet},
$\blank \in \Gamma \setminus \Sigma$ is the \emph{blank symbol}, $\delta \colon Q \times \Gamma \to Q \times \Gamma \times \{L,N,R\}$
is the \emph{transition function}, and $\qin, \qyes, \qno \in Q$ are the unique \emph{starting}, \emph{accepting} and \emph{rejecting states}. 

A \emph{configuration} of a DTM is a tuple $(q,h,t)$ where $q \in Q$ is the current state, $h \in \mathbb{N}$ is the head position, and $t \colon \mathbb{N} \to \Gamma$ is the tape content. The \emph{initial configuration} on input word $w \in \Sigma^*$ is given by $(\qin, 0, t)$ with $t(i) = w_i$
if $i <|w|$ and $t(i) = \blank$ otherwise. The unique \emph{accepting} and \emph{rejecting configurations} are given by $(\qyes, 0, t)$, resp.\
$(\qno, 0, t)$ where $t(i) = \blank$ for all $i$ in both cases.

A configuration $C = (q', h', t')$ is the, necessarily unique, \emph{successor configuration} of $(q,h,t)$ if 
(\textsc{i}) $q \notin \{\qyes, \qno\}$, 
(\textsc{ii}) $\delta(q,t(h)) = (q', \gamma, D)$ for some $\gamma \in \Gamma, D \in \{L,N,R\}$, 
(\textsc{iii}) $t'(i) = \gamma$ if $i = h$ and $t'(i) = t(i)$ otherwise, and
(\textsc{iv}) $h' = h-1$ if $D=L$ and $h > 0$, $h' = h$ if $D=N$, and $h' = h+1$ if $D = R$. A (partial) \emph{computation} of $\mathcal{M}$ on input $w$ is
a finite or infinite sequence of configurations $C_0,C_1,\dotsc$ where $C_0$ is the initial configuration of $\mathcal{M}$ on $w$, and 
$C_{i+1}$, if it exists, is the successor configuration of $C_i$. Such a computation is \emph{maximal} if it is either infinite or its 
last configuration is the accepting or the rejecting configuration. 
Note that each $\mathcal{M}$ has exactly one maximal computation for each $w$, whence from now on we talk about \emph{the} computation
of $\mathcal{M}$ on $w$. 
We say that $\mathcal{M}$ \emph{accepts} $w$ if its unique maximal computation on $w$ ends with the accepting configuration, and we write $L(\mathcal{M})$ 
for the set of words accepted by $\mathcal{M}$. Conversely, $\mathcal{M}$ \emph{rejects} $w$ if its unique maximal computation on $w$ ends in the
rejecting configuration or if it is infinite. In the latter case, we say that the computation \emph{diverges}.

We say that a non-diverging computation $C_0,\dotsc,C_k$ on input some $w$ consumes space $n$, written $\text{space}_\mathcal{M}(w) = n$, if 
$n = 1 + \max\{h_i \mid C_i = (q_i, h_i, t_i)\}$. Obviously, if the head never advances beyond position $n-1$, then $t_i(j) = \blank$ for all
$0 \leq i \leq k$ and $j > n-1$.
Let $f \colon \mathbb{N} \to \mathbb{N}$ be a function. We say that $\mathcal{M}$ is $f$-space-bounded if $\mathcal{M}$ has no diverging
computations on any input and, for all $n$, we have $\max\{\text{space}_\mathcal{M}(w) \mid w \in \Sigma^n\} \leq f(n)$. We say that $\mathcal{M}$
is \emph{$k$-fold-exponential space bounded} if there is a polynomial $p(n)$ such that $\mathcal{M}$ is $\twr{p(n)}{k}$-space-bounded.

\subsection{Higher-Order Logic with Partial Fixpoints}

In order to keeps things notationally simple, we restrict ourselves to the class of labeled transition systems (LTS), or labelled
graphs. Let $\Prop = \{p,q,\dotsc\}$ be a set of \emph{propositions} and let $\Act = \{a,b,\dotsc\}$ be a set of \emph{actions} or
\emph{transition relation} or \emph{edge relations}. An LTS is a tuple $\Transsys = (\States, \Actions, \Label)$ 
where $\States = \{\state, \statealt, \dotsc\}$ 
is a finite, nonempty set of \emph{states}, $\Actions \subseteq \States \times \Act \times \States$ is the \emph{transition relation}
and $\Label \colon \States \to 2^{\Prop}$ labels each state by the set of propositions valid in it. 
We write $\Transition{\state}{a}{\statealt}$ instead of $(\state, a, \statealt) \in \Actions$.

\paragraph*{Types.}
The set of types is defined via the grammar 
\[
\tau,\tau_1,\dotsc,\tau_n \Coloneqq \grtype \mid \tau_1,\dotsc,\tau_n \mid (\tau)
\]
where $\grtype$ is the \emph{ground type} or \emph{type of individuals} of \emph{order} $\ord{\grtype} = 1$, $\tau_1,\dotsc,\tau_n$ is a 
\emph{compound type} of order $\ord{\tau_1,\dotsc,\tau_n} =  \max\{\ord{\tau_1},\dotsc,\ord{\tau_n}\}$, and where $(\tau)$
is a \emph{set type} of order $\ord{(\tau)} = 1 + \ord{\tau}$.\footnote{Compound type and set type are often combined into a single
``set of tuples'' type. We use separate operators here for ease of notation, but the results of the paper do not depend on that.} 

Given an LTS $\Transsys = (\States, \Actions, \Label)$, the semantics $\semm{\tau}{\Transsys}$ of a type $\tau$ is given by
\[\semm{\grtype}{\Transsys} = \States, \qquad\qquad \semm{\tau_1,\dotsc,\tau_n}{\Transsys} = 
\semm{\tau_1}{\Transsys} \times \dotsb \times \semm{\tau_n}{\Transsys}\qquad \qquad  \semm{(\tau)}{\Transsys} = 2^{\semm{\tau}{\Transsys}}. \] 
We often compress compound and set types by writing e.g., $(\tau_1,\dotsc,\tau_n)$.

The following is straightforward to prove by induction on the structure of types.
\begin{lemma}
\label{lem:typesize}
For any $\tau$ of order $k$ and any LTS $\Transsys$, with state set $\States$, $|\semm{\tau}{\Transsys}|$ is $k-1$-fold exponential in $|\States|$. 
\end{lemma}

Given an LTS as above, and some $f \colon \semm{\tau}{\Transsys}\to\semm{\tau}{\Transsys}$ for $\tau$ of order at least $2$,
we define its \emph{partial fixpoint} \pfpx{f} via
\[
\pfpx{f} = \left\{\begin{aligned} F^i &, \text{ if } i \text{ exists such that } F^i = F^{i+1} \\ \emptyset &, \text{ otherwise,}  \end{aligned}\right.
\]
where $F^0 = \emptyset$ and $F^{i+1} = f(F^i)$. By an obvious counting argument, if there is $i$ such that a nontrivial partial
fixpoint exists, then there already is one bounded by $|\semm{\tau}{\Transsys}|$, which, by Lem.~\ref{lem:typesize}, is $k-1$-fold exponential in $|\States|$ for $\tau$
of order $k$.

\paragraph*{Syntax.}
Let $\vars = \{X,Y,\dotsc \}$ be a set of typed variables, tacitly assumed to contain infinitely many variables for each type. 
%Variables of type $\grtype$ are denoted by lower case letters, variables of compound types are denoted by upper case letters.
%
The set of \HOPFP-formulas is defined by the grammar
\[
\varphi  \Coloneqq \mytrue \mid p(X) \mid a(X,Y) \mid X(Y_1,\dotsc,Y_n) \mid \neg \varphi \mid \varphi \vee \varphi \mid
\exists (X\colon\tau).\ \varphi \mid (\pfp (X\colon\tau).\ \varphi)(Y_1,\dotsc,Y_n)
\]
where $p \in \Prop$, $a \in \Act$, and $X,Y_1,\dotsc,Y_n,$ are variables. A formula $\varphi$ is \emph{well-formed} if the following
are true for $\varphi$: (\textsc{i}) The variables in terms of the form $p(X)$ or $a(X,Y)$ are of type $\grtype$, and
(\textsc{ii}) in a term of the form $X(Y_1,\dotsc,Y_n)$ or $(\pfp (X\colon\tau).\ \varphi)(Y_1,\dotsc,Y_n)$, 
the variable $X$ has type $(\tau_1,\dotsc,\tau_n)$ if $Y_i$ has type $\tau_i$ for $1\leq i \leq n$. If they are not important,
we omit type annotations of the form $(X \colon \tau)$, and we use compressed notation such as $\exists (X,Y,Z\colon \grtype).\, \varphi$
or $\exists (X,Y,Z \colon \tau, \tau', \tau'')$ where appropriate.

Other derived operators such as $\wedge, \rightarrow, \forall, \myfalse$ etc.\ can be added in the usual way. The notions
of subformula, formula size etc.\ are also standard. Free and bound variables are defined as usual, with $X$ being
a bound variable in $(\pfp (X\colon\tau).\ \varphi)(Y_1,\dotsc,Y_n)$. We use notation such as $\varphi (X\colon \tau, Y\colon \tau')$ etc.\
to communicate the names and types of the free variables of a formula, with shorthands as above used if appropriate.

We say that $\varphi$ has order $k$ if the highest order of a variable that occurs freely or as $X$ in a formula of the 
form $\exists X.\ \psi$ is at most $k$, and the highest order of a variable $X$ in a subformula of the form
$(\pfp (X\colon\tau).\ \varphi)(Y_1,\dotsc,Y_n)$ is at most $k+1$. We write \HOP{k} for the collection of all formulas of order at most $k$.

\paragraph*{Semantics.}
Let $\Transsys = (\States, \Actions, \Label)$ be an LTS. A variable assignment $\eta$ is a function that assigns, to each
variable $X \in \vars$ of type $\tau$, an element of $\semm{\tau}{\Transsys}$. Given some $X$ of type $\tau$ and some
$f \in \semm{\tau}{\Transsys}$, the update $\eta[X \mapsto f]$ is defined as $\eta[X\mapsto f](X) = f$ and $\eta[X\mapsto f](Y) = \eta(Y)$
if $Y \not = X$.

The semantics of a \hop\ formula is defined as follows:
\begin{align*}
\Transsys, \eta \models \mytrue &\text{, always} \\
\Transsys, \eta \models p(X) &\text{, iff } p \in \Label(\eta(x)) \\
\Transsys, \eta \models a(X,Y) &\text{, iff } \Transition{\eta(X)}{a}{\eta(Y)} \\
\Transsys, \eta \models X(Y_1,\dotsc,Y_n) &\text{, iff } (\eta(Y_1),\dotsc,\eta(Y_n)) \in \eta(X) \\
\Transsys, \eta \models \neg \varphi &\text{, iff } \Transsys, \eta \not\models  \varphi\\
\Transsys, \eta \models \varphi_1 \vee \varphi_2 &\text{, iff } \Transsys, \eta \models \varphi_1 \text { or } \Transsys, \eta \models \varphi_2 \\
\Transsys, \eta \models \exists (X\colon\tau). \varphi &\text{, iff } \text{ there is } f \in \semm{\tau}{\Transsys} \text{ s.t.\ } \Transsys, \eta[X \mapsto f] \models \varphi \\
\Transsys, \eta \models (\pfp (X\colon\tau) \varphi)(Y_1,\dotsc,Y_n) &\text{, iff } (\eta(Y_1),\dotsc,\eta(Y_n)) \in \pfpx{\varphi^\Transsys_\eta}
\end{align*}
where $\varphi^\Transsys_\eta$ is the function that maps $g \in \semm{\tau}{\Transsys}$ to 
\[\{(g_1,\dotsc,g_n) \in 2^{\semm{\tau_1}{\Transsys} \times \dotsb \times \semm{\tau_n}{\Transsys}} \mid 
\Transsys,\eta[X \mapsto g, Y_1 \mapsto g_1,\dotsc,Y_n\mapsto g_n] \models \varphi \} \in \semm{\tau}{\Transsys}
\]
if $\tau = (\tau_1,\dotsc,\tau_n)$.
                         
\subsection{Queries}

Let $\Prop$ and $\Act$ be fixed. A (boolean) \emph{query} (over $\Prop$ and $\Act$) is a function $\mathcal{Q}$ that maps, to each finite LTS 
$\Transsys = (\States, \Actions, \Label)$ a truth value, i.e., either true or false. Alternatively, such a boolean query is just
a set of finite LTS over (over $\Prop$ and $\Act$), which we shall identify with $\mathcal{Q}$.

A closed \HOPFP formula $\varphi$ naturally defines a query via 
\[
\mathcal{Q}_\varphi = \{\Transsys \mid \Transsys  \models \varphi \}.
\]

Conversely, queries can be decided by space-bounded Turing machines. For this, the machine receives the LTS in question 
as an input, and either accepts or rejects. The LTS has to be encoded into some word of the input alphabet for this. Naturally,
this introduces a total order on the set of states of the LTS. It is known that e.g., the expressive power of first-order
logic increases in the presence of an order (this is a classic exercise when introducing Ehrenfeucht-Fra\"issé games).
However, order is not an issue in our setting. The classical first-order characterization due to Abiteboul and Vianu is explicitly 
restricted to ordered structures, and characterizations for logics beyond existential second-order logic can be done with an
order in mind, as existential second-order logic is strong enough to simply guess an order. 
This includes \hopfp{k} for $k \geq 2$, i.e., the topic of this paper.

Given an LTS $\Transsys$,  let $\langle \Transsys \rangle$
be some form of polynomial encoding of $\Transsys$ into a given input alphabet $\Sigma$, e.g., using adjacency matrices or the like.
We say that a Turing machine $\mathcal{M}$ decides a query $\mathcal{Q}$ if $\mathcal{M}$ halts on any input of the form $\langle \Transsys\rangle$,
where $\Transsys$ is necessarily finite, and accepts exactly those codings where $\Transsys \in \mathcal{Q}$. 
A query is a $k$-\ExpSpace-query if there is $\mathcal{M}$ that is $k$-fold-exponential space bounded and decides $\mathcal{Q}$.

We now say that a logic $\mathcal{L}$ captures a complexity class $\mathcal{C}$ over a class of structures (LTS) $\mathcal{S}$ if,
for each $\mathcal{L}$-query there is a $\mathcal{C}$-query that yields the same set when restricted to $\mathcal{S}$, and
vice versa. 

\begin{remark}
Non-boolean queries are quite common in e.g., the field of database theory. A $d$-query is then not a function that maps an LTS to a
truth value, but rather one that maps an LTS and a $d$-tuple of states to a truth value, or, equivalently, maps every LTS to 
a set of $d$-tuples. On the logical side, one now deals with formulas with free first-order variables.
We choose to stick to boolean queries here in order to avoid the extra coding required to get said free variables encoded into 
DTM.
\end{remark}

\subsection{Vardi's Characterization of PSPACE}

We briefly sketch the classical result due to Vardi \cite{DBLP:conf/stoc/Vardi82} that first-order logic with partial fixpoints, i.e.,
$\hopfp{1}$, captures \PSpace over the class of ordered LTS. One direction is rather straightforward since
first-order queries can be evaluated in polynomial \emph{time}, the individual stages of a partial fixpoint
only take polynomial space, and the next stage can be computed from the previous
one also in polynomial time. Since such a partial fixpoint either does not stabilize, or stabilizes after at most
exponentially many iterations, it is sufficient to keep a counter for the number of iterations, which takes
polynomially many bits if it is encoded in binary. 

For the other direction, let $\mathcal{M} = (Q, \Sigma, \Gamma, \blank, \delta, \qin, \qyes, \qno)$ be a $p(n)$-space-bounded DTM
that decides a query $\mathcal{Q}$ over LTS, i.e., it accepts those $w = \langle \Transsys \rangle$ such that $\Transsys \in \mathcal{Q}$.

Since $\mathcal{M}$ is $p(n)$-space-bounded, the tape contents 
and the head position of each configuration of a computation of $\mathcal{M}$ on an input of length $n$ can be represented
by a number of at most $p(n)$ and a word of length $p(n)$ over the tape alphabet of $\mathcal{M}$.  The proof rests
on three key observations:
\begin{itemize}
 \item In sufficiently large, ordered LTS, a configuration of $\mathcal{M}$ can be represented as a second-order relation of
       sufficient arity,
 \item the operator that computes from such a representation of a configuration its successor configuration, if it exists,
       can be expressed as a first-order formula, and
 \item the initial and accepting configurations can be pinned down using first-order logic.
\end{itemize}
The capturing result is then obtained by observing that $\mathcal{M}$ accepts its input $w$, derived from $\Transsys$, 
iff the partial fixpoint obtained
by feeding a representation of the initial configuration into the operator mentioned above is nonempty and contains exactly
a representation of the accepting configuration.

\section{\texorpdfstring{\hopfp{k}}{HO-k+PFP}-Queries are in \texorpdfstring{$k-1$-\ExpSpace}{k-1-EXPSPACE}}
\label{sec:upper}
% !TEX root =  main.tex

We begin with the simpler part of the capturing result. We will show that queries definable in \hopfp{k} can be
evaluated using at most $(k-1)$-fold exponential space. This does not even need any special tricks. 
Alg.~\ref{alg:eval} essentially just computes the semantics of an \hopfp{k} query $\varphi$ w.r.t.\ an LTS
$\Transsys$ and a variable evaluation $\eta$, i.e., it decides whether or not $\Transsys, \eta \models \varphi$
holds. 

\begin{algorithm}
\begin{algorithmic}[1]
\Procedure{Eval}{$\Transsys,\eta,\varphi$}
  \Cases{$\varphi$}
    \Case{$\mytrue$} \Return \texttt{true}
    \Case{$p(X)$} \Return $p \in \ell(\eta(X))$
    \Case{$a(X,Y)$} \Return $\Transition{\eta(X)}{a}{\eta(Y)}$
    \Case{$X(Y_1,\dotsc,Y_n)$} \Return $(\eta(Y_1),\dotsc,\eta(Y_n)) \in \eta(X)$    
    \Case{$\neg \psi$} \Return $\neg \textsc{Eval}(\Transsys,\eta,\psi)$    
    \Case{$\psi_1 \vee \psi_2$} \Return $\textsc{Eval}(\Transsys,\eta,\psi_1) \vee \textsc{Eval}(\Transsys,\eta,\psi_2)$    
    \Case{$\exists (X\colon\tau).\ \psi$}
      \ForAll{$f \in \sem{\tau}{\Transsys}{}$}
        \If{$\textsc{Eval}(\Transsys,\eta[X \mapsto f],\psi)$}
          \State \Return \texttt{true}
        \EndIf
      \EndFor
      \State \Return \texttt{false}
    \Case{$(\pfp (X\colon(\tau_1,\dotsc,\tau_k)).\ \psi)(Y_1,\dotsc,Y_n)$}
      \State{$f \gets \emptyset$} % \Comment{least element in $\sem{\tau}{\Transsys}{}$}
      \State{$\mathit{cnt} \gets 0$}
      \While{$\mathit{cnt} < |\sem{(\tau_1,\dotsc,\tau_k)}{\Transsys}{}|$}
        \State{$f' \gets f$}
				\State{$f \gets \emptyset$}
        \ForAll{$(M_1,\dotsc,M_k) \in \semm{\tau_1}{\Transsys} \times \dotsb \times \semm{\tau_k}{\Transsys}$}
          \If{$\Transsys, \eta[Y_1 \mapsto M_1,\dotsc,Y_k \mapsto M_k, X \mapsto f'] \models \psi$}
					  \State{$f \gets f \cup \{ (M_1,\dotsc,M_k)\}$}
						\EndIf
        \EndFor
        \If{$f = f'$}
          \State{\Return $(\eta(Y_1),\dotsc,\eta(Y_n)) \in f$}
        \EndIf
        \State{$\mathit{cnt} \gets \mathit{cnt} + 1$}
      \EndWhile
      \State \Return \texttt{false}
  \EndCases
\EndProcedure
\end{algorithmic}
\caption{Evaluating \hopfp{k} queries in $(k-1)$-fold exponential space.}
\label{alg:eval}
\end{algorithm}

\begin{theorem}
\label{thm:upper}
Let $k \ge 2$. Evaluating an \hopfp{k} query is in $(k-1)$-\ExpSpace.
\end{theorem}

\begin{proof}
It is not hard to see that algorithm \textsc{Eval} correctly evaluates an \hopfp{} query, as it closely follows
the semantics of \hopfp{}. It remains to be seen that the space needed by this procedure is bounded by a function
that is at most $(k-1)$-fold exponential in the size of the underlying $|\Transsys|$ and $|\varphi|$.

First note that the recursion depth in \textsc{Eval} is bounded by $|\varphi|$. Hence, it suffices to check that
the space needed within each recursive call is bounded in this way. It is only the last two cases in which this
may not be obvious. So consider the case of $\varphi = \exists (X\colon\tau).\ \psi$. Enumerating all elements
of $\sem{\tau}{\Transsys}{}$ requires space for one of these elements, plus space either for a counter to abort
the enumeration after all elements have been constructed, or for a second of these elements in case the enumeration
can construct, from one of these elements, a uniquely determined successor (in a lexicographic ordering for instance).
In both cases, the space needed is logarithmic in $|\sem{\tau}{\Transsys}{}|$ which is at most $(k-1)$-fold 
exponential in $|\Transsys|$ according to Lemma~\ref{lem:typesize}.

The argument for the last case of $\varphi = (\pfp (X\colon(\tau_1,\dotsc,\tau_k)) \psi)(Y_1,\dotsc,Y_n)$ is similar.
We write $\tau$ for $(\tau_1,\dotsc,\tau_k)$. Note that the order of $\tau$ may be up to $k+1$, so 
$|\semm{\tau}{\Transsys}$ is $k$-fold exponential in $\States$.
 The space needed to evaluate the partial fixpoint formula is determined by a counter with values up to $|\semm{\tau}{\Transsys}|$
and by the two elements $f,f' \in \semm{\tau}{\Transsys}$. Using binary coding, the space needed for the counter is 
logarithmic in $|\semm{\tau}{\Transsys}$, and individual elements of $\semm{\tau}{\Transsys}$ take $k-1$-fold
exponential space, too. Hence, the space needed in this case is also at most $(k-1)$-fold exponential in $|\Transsys|$.  
\end{proof}

\section{\texorpdfstring{$k-1$-\ExpSpace}{k-1-EXPSPACE}-Queries are Expressible in \texorpdfstring{\hopfp{k}}{HO-k+PFP}}
\label{sec:lower}

\subsection{Ordering Higher-Order Relations}

Since we want to encode runs of $k$-fold-exponentially space-bounded Turing machines into formulas of \hopfp{k+1}, we have to be able
to encode the tape contents of the Turing machine in question. For such a space-bounded machine, the tape can be represented
by a $\Gamma$-word of $k$-fold exponential length, where $\Gamma$ is the tape alphabet of the machine in question. Hence,
we have to be able to somehow represent such a large word or, in other words, we must be able to count to large numbers. 

Let $p(n)$ be a polynomial, for the time being one of the form $n^c$ for some $c \geq 2$.
Let $\Act$ contain a relation $<$, and let the types $\tau_0,\dotsc,\tau_k$ be the types defined via
$\tau_1 = \grtype^c$, i.e. $\grtype \times \dotsb \times \grtype$ with $c$ many repetitions of $\grtype$, and $\tau_{i+1} = (\tau_i)$.
We define formulas $\varphi_<^{1}$, $\varphi_<^2$ and $\varphi_<^{i+1}$ for $i \geq 2$ via:
\begin{align*}
\varphi_<^1(X_1,\dots,X_c, Y_1,\dotsc,Y_c \colon\grtype) &= \bigvee_{i=1}^c <(X_i,Y_i) \wedge \bigwedge_{j =  1}^{i-1} \neg <(Y_j,X_j) \\
\varphi_<^{2}(X,Y \colon \tau_1) &=  \exists (Z_1,\dotsc,Z_c\colon\grtype).\, Y(Z_1,\dotsc,Z_{c}) \wedge \neg X(Z_1,\dotsc,Z_{c})  \\
                                       &\;\qquad \qquad \wedge \forall (Z'_1,\dotsc,Z'_{c}\colon\grtype).\, \big(\varphi_<^1(Z'_1,\dotsc,Z'_{c},Z_1,\dotsc,Z_{c}) \\
                                      & \qquad \qquad \qquad \qquad \qquad \quad
																			\rightarrow X(Z'_1,\dotsc,Z'_{c}) \rightarrow Y(Z'_1,\dotsc,Z'_{c}) \\
\varphi_<^{i+1}(X,Y \colon \tau_{i+1}) &=  \exists (Z\colon\tau_i).\, Y(Z) \wedge \neg X(Z) \wedge \forall (Z'\colon\tau_i).\, \varphi_<^i(Z',Z))  \rightarrow \big( X(Z') \rightarrow Y(Z')\big) 
\end{align*}
Here, $\varphi_<^1$ defines a total order on $\semm{\grtype^c}{\Transsys}$ via the lexicographical ordering induced by $<$.
For $i \geq 1$, the formula $\varphi_<^{i+1}$ then totally orders $\semm{\tau_{i+1}}{\Transsys}$ via lexicographical ordering of sets w.r.t.\
the membership of elements of $\tau_i$.
\begin{lemma}
\label{lem:order}
Let ${<} \in \Act$ and let $\Transsys = (\States, \Actions, \Label)$ be an LTS over 
$\Act$ and some $\Prop$ such that $<$ is a total order on $\States$. 
Let $\tau_k$ for $k \geq 1$ be defined as above. Then the following are true for all $k \geq 1$:
(\textsc{i}) $|\semm{\tau_k}{\Transsys}| = \twr{|\States|}{k-1}$,
(\textsc{ii}) $\varphi_<^{k}$ defines a total order on $\semm{\tau_{k}}{\Transsys}$.
\end{lemma}
Additionally, let $\varphi^{1}_=, \varphi_=^{i}$ and $\varphi^{1}_{\text{succ}}, \varphi^{i}_{\text{succ}}$ for $i > 1$ be defined as
\begin{align*}
 \varphi^1_=(X_1,\dotsc,X_c, Y_1,\dotsc,Y_c \colon \grtype) &= \neg \varphi_<^1(X_1,\dotsc,X_c, Y_1,\dotsc,Y_c) \wedge \neg \varphi_<^1(Y_1,\dotsc,Y_c, X_1,\dotsc,X_c) \\
 \varphi^{i}_=(X,Y\colon \tau_i) &= \neg \varphi_<^i(X,Y) \wedge \neg \varphi_<^i(Y,X) \\
 \varphi^1_{\text{succ}} (X_1,\dotsc,X_c, Y_1,\dotsc,Y_c \colon \grtype)  &=  \varphi_<^1((X_1,\dotsc,X_c, Y_1,\dotsc,Y_c) \wedge \forall (Z_1,\dotsc,Z_c \colon \grtype).\,  \varphi_=^1(X_1,\dotsc,X_c, Y_1,\dotsc,Y_c) \\ & \qquad \qquad \qquad \rightarrow (\varphi_=^1(X_1,\dotsc,X_c, Z_1,\dotsc,Z_c) \vee \varphi_<^1(Z_1,\dotsc,Z_c,X_1,\dotsc,X_c)) \\
 \varphi^{i}_{\text{succ}}(X,Y\colon \tau_i) &= \varphi_<^i(X,Y) \wedge \forall(Z\colon\tau_i).\, \varphi_<(Z,Y) \rightarrow \varphi_=(X,Z) \vee \varphi_<(Z,X)  
\end{align*}
expressing equality between elements of $\semm{\tau_i}{\Transsys}$ or the fact that the second argument is the immediate successor of the first one w.r.t. the total order induced by $\varphi_<^i$.

Finally, for each $j \in \mathbb{N}$ and $i > 1$, define the formulas $\varphi^1_{=j}$, $\varphi^{i}_{=j}$ via
\begin{align*}
 \varphi^1_{=0}(X_1,\dotsc,X_c\colon \grtype) &= \forall (Y_1,\dotsc,Y_c \colon\grtype).\, \varphi_<^1(X_1,\dotsc,X_c,Y_1,\dotsc,Y_c) \vee \varphi_=^1(X_1,\dotsc,X_c,Y_1,\dotsc,Y_c) \\
\varphi^1_{=j+1}(X_1,\dotsc,X_c\colon \grtype) &= \exists (Y_1,\dotsc,Y_c \colon\grtype).\, \varphi_{=j}^1(Y_1,\dotsc,Y_c) \wedge \varphi_{\text{succ}}^1(X_1,\dotsc,X_c,Y_1,\dotsc,Y_c) \\
\varphi^{i}_{=0}(X\colon\tau_{i}) &= \forall (Y\colon\tau_i).\, \varphi_<^i(X,Y) \vee \varphi_=^{i}(X,Y) \\
\varphi^{i}_{=j+1}(X\colon\tau_{i}) &= \exists (Y\colon\tau_i).\, \varphi_{=j}^i(X,Y) \wedge \varphi_{\text{succ}}^{i}(X,Y) 
\end{align*}
where $\varphi_{=j}^0$ and $\varphi_{=j}^i$ express that $(X_1,\dotsc,X_c)$, resp.\ $X$ is the $j+1$st element of the total order induced
by $\varphi_<^{0]}$, resp.\ $\varphi_<^i$, if such an element exists. Clearly, the size of these formulas is linear
in $j$.

\subsection{The Reduction}

Let $k \geq 1$ and let $\mathcal{M} = (Q, \Sigma, \Gamma, \blank, \delta, \qin, \qyes, \qno)$ be a $\twr{k}{p(n)}$-space-bounded DTM
that decides a query $\mathcal{Q}_\mathcal{M}$ over ordered LTS, i.e., it accepts those $w = \langle \Transsys\rangle$
for which $\Transsys \in \mathcal{Q}$. W.l.o.g.\ $p(n) = c \cdot n^{c-1}$ for some $c$, whence also 
for $n \geq c$ we have $p(n) \leq n^c$. We also assume that $\mathcal{M}$ rejects all inputs that do not encode an LTS ordered by a relation $<$.

We have to build a \hopfp{k+1} formula $\varphi(X_1,\dotsc,X_d)$ such that 
$\Transsys  \models \varphi$ iff $\Transsys \in \mathcal{Q}_{\mathcal{M}}$
and $\Transsys$ is ordered by $<$.

\paragraph*{Encoding Configurations.}

Let $\tau= \grtype, \tau_{k+1}, \tau_{k+1}, \grtype$. Let $\Transsys$ be an LTS ordered by $<$ such that its state set satisfies
$|\States| \geq \max\{c, |Q|, |\Gamma| \}$. Hence, $|\States|^c \geq p(|\States|)$. 
W.l.o.g. $Q$ and $\Gamma$ are ordered, i.e., $Q = \{q_0,\dotsc,q_{|Q|-1}\}$ and $\Gamma = \{\gamma_0,\dotsc,\gamma_{|\Gamma|-1}\}$.
Since $\States$ is ordered by $<$, for each $q_i \in Q$ and for each $\gamma_j \in \Gamma$, there are unique 
states $\state_q$ and $\state_\gamma$, given as the $i+1$st, resp.\ $j+1$st states in the total order $<$.
An element of $\semm{\tau}{\Transsys}$ has the form $(\state, H, I, \state')$ with 
$\state, \state' \in \semm{\grtype}{\Transsys} = \States$ and $H,I \in \semm{\tau_{k+1}}{\Transsys}$. 
\begin{definition}
\label{def:code-config}
Let $M \in \semm{(\tau)}{\Transsys}$. We say that $M$ encodes a configuration $C = (q, h, t)$ of $\mathcal{M}$ if the following are true:
\begin{enumerate}
\item For all $(\state,H,I,\state') \in M$, we have that $\state = \state_q$.
\item For all $(\state,H,I,\state'), (\statealt,H',I',\statealt') \in M$, we have $\state = \statealt$ and $H = H'$ and $H$ is the $h+1$st element
      in the total order induced by $\varphi_<^{k+1}$.
\item For each $I \in \semm{\tau_{k+1}}{\Transsys}$, there is exactly one tuple of the form $(\state, H,I,\state')$ in $M$.
\item If $j \leq \twr{p(|\States|)}{k}$, if $I$ is the $j+1$st element in the total order induced by $\varphi_<^{k+1}$, and if
     $(\state, H,I,\state') \in M$, then $\state'= \state_\gamma$ for some $\gamma \in \Gamma$ and $t(j) = \gamma$.
\end{enumerate}
\end{definition}
The intuition here is the following: Since all tuples in $M$ agree on $\state_q$ and $H$, this uniquely determines $q$ and $h$. 
Moreover, since for each $I \in \semm{\tau_{k+1}}{\Transsys}$, there is exactly one tuple of the form $(\state, H,I,\state')$ in $M$,
this defines a function $\semm{\tau_{k+1}}{\Transsys} \to \Gamma$, and since $\semm{\tau_{k+1}}{\Transsys}$ is linearly
ordered via $\varphi_<^{k+1}$ and has cardinality $\twr{p(|\States|)}{k}$ due to Lem.~\ref{lem:order}, this yields
a function $\{0,\dotsc,\twr{|\States|^c}{k}-1\} \to \Gamma$.
Since $\mathcal{M}$ is $\twr{p(n))}{k}$-space-bounded, all configurations of a run of $\mathcal{M}$ on 
input $\langle \Transsys, (\state_1,\dotsc,\state_d)\rangle$ have a head position less than $\twr{p(|\States|)}{k} \leq 
|\semm{\tau_{k+1}}{\Transsys}|$ and, consequently, all tape cells of such a configuration with index at least $\twr{p(|\States|)}{k}$
must contain $\blank$.
Hence, such a set in $\semm{(\tau)}{\Transsys}$ can encode any configuration $\mathcal{M}$ may enter during its run on
input $\langle \Transsys, (\state_1,\dotsc,\state_d)\rangle$.

Now let $w = \langle \Transsys\rangle$. 
Consider the \hopfp{k+1} formula
\begin{align*}
\varphi^w_{\text{init}}(Y_q\colon \grtype, H \colon \tau_{k+1}, I \colon \tau_{k+1}, Y_\gamma \colon \grtype) &= 
   \varphi_{=0}^{k+1}(H) \wedge Y_q = \state_{\qin} \wedge  \bigwedge_{i = 0}^{|w|-1} \varphi_{=i}^{k+1}(I) \rightarrow Y_\gamma = \state_{w_i} \\
				 &	\quad \wedge \exists (Z\colon\tau_{k+1})\ \varphi_{=|w|-1}^{k+1}(Z) \wedge \varphi_<^{k+1}(Z,I) \rightarrow Y_\gamma = \state_\blank.
\end{align*}
It is of polynomial size and expresses that the tuple encoded in the variables $Y_q, H, I, Y_\gamma$ is in the unique set that encodes
the initial configuration
of $\mathcal{M}$ on input $w$. We use shorthand such as $Y_q = \state_{\qin}$ to abbreviate $\varphi_{=j}^1(Y_q)$ if $\qin$ is the $j+1$st
state w.r.t. $<$ on $\States$.

\paragraph*{The Partial Fixpoint.}
Consider the formula
\begin{align*}
\psi_{\text{trans}}(X, Y_q, H, I, Y_\gamma) &= \exists (Y'_{q}, H',I',Y'_{\gamma}\colon\grtype,\tau_{k+1},\tau_{k+1},\grtype).\ \\
                                            &\quad \exists (Y''_{q}, H'',I'',Y''_{\gamma}\colon\grtype,\tau_{k+1},\tau_{k+1},\grtype).\ \\
                                            &\quad X(Y'_{q}, H',I',Y'_{\gamma})
																						\wedge X(Y''_{q}, H'',I'',Y''_{\gamma}) 
																						\wedge \varphi_=^{k+1}(H',I'') \wedge \varphi_=^{k+1}(I, I')\\
																						& \wedge \neg \varphi_=^{k+1}(H,I) \rightarrow \varphi_=^{1}(Y_\gamma,Y'_\gamma) \\
																						& \wedge \bigwedge_{(q',\gamma, q'',\gamma', L) \in \delta} Y'_q = \state_{q'} \wedge Y'_{\gamma} = \state_\gamma \rightarrow Y = \state_{q''} \wedge \varphi_{\text{succ}}^{k+1}(H, H') \wedge \varphi_=^{k+1}(H,I) \rightarrow Y_\gamma = \state_{\gamma'} \\
																						& \wedge \bigwedge_{(q',\gamma, q'',\gamma', N) \in \delta} Y'_q = \state_{q'} \wedge Y'_{\gamma} = \state_\gamma \rightarrow Y = \state_{q''} \wedge \varphi_{=}^{k+1}(H, H') \wedge \varphi_=^{k+1}(H,I) \rightarrow Y_\gamma = \state_{\gamma'} \\
																						& \wedge \bigwedge_{(q',\gamma, q'',\gamma', R) \in \delta} Y'_q = \state_{q'} \wedge Y'_{\gamma} = \state_\gamma \rightarrow Y = \state_{q''} \wedge \varphi_{\text{succ}}^{k+1}(H', H) \wedge \varphi_=^{k+1}(H,I) \rightarrow Y_\gamma = \state_{\gamma'}
\end{align*}
where by abuse of syntax we write $(q',\gamma, q'',\gamma', L) \in \delta$ instead of $\delta(q',\gamma) = (q'', \gamma',L)$ etc.
\begin{lemma}
\label{lem:trans}
Assume that $M \in \semm{(\tau)}{\Transsys}$ with $\tau = (\grtype, \tau_{k+1},\tau_{k+1},\grtype)$ as before encodes some configuration
$C$ of the computation of $\mathcal{M}$ on input $w$, and assume that $M'$ of the same type encodes the successor configuration of $C$.

Let $\state, \state' \in \States$ and $M_H, M_I \in \semm{\tau_{k+1}}{\Transsys}$. Then 
\[\Transsys, \eta[X \mapsto M, Y_q \mapsto \state, H \mapsto M_H, I \mapsto I_H, Y_\gamma \mapsto \state'] \models \psi_{\text{trans}}
\quad\text{ iff }\quad (\state, M_H, M_I, \state') \in M'.\]

\end{lemma}
The intuition here is that  $\psi_{\text{trans}}$ defines the encoding of a successor of some configuration $C$ from the encoding of $C$ itself.
The first existential quantifier requires the existence of a tuple in $X$ that encodes the value of the tape at the same position as 
the new tuple will, i.e., they both must have the same third component, and the second quantifier requires the existence of a tuple that encodes
the content of the tape at the head position. The third line enforces these properties. The fourth line fixes tape contents not under the head.
The last three lines, separated for the ease of notation, enforce that both the state transition and the new content of the tape at the old
head position obey the transition function.

We now have the required machinery to encode a computation of $\mathcal{M}$ on input $w$ into \hopfp{k+1}. Let
\begin{align*}
\psi(X,Y_q,H,I,Y_\gamma) &= \varphi_{\text{trans}}(X, Y_q, H, I, Y_\gamma)   \\ 
                 & \qquad \qquad \vee\forall(Y'_q, H', I', Y'_\gamma \colon \grtype, \tau_{k+1},\tau_{k+1},\grtype).\, \neg X(Y'_q,H',I',Y'_\gamma) 
                                                  \wedge \varphi_{init}^w(Y_q, H, I, Y_\gamma) \\
\varphi_{\mathcal{M}} &=  \varphi^< \wedge \exists (Y'_q, H', I', Y'_\gamma \colon \grtype, \tau_{k+1},\tau_{k+1},\grtype).\ Y'_q = \state_{\qyes} \wedge (\pfp (X\colon(\tau).\, \psi)(Y'_q,H',I',Y'_\gamma).
\end{align*}
where $\varphi^<$ expresses that $<$ is a total order.
\begin{lemma}
\label{lem:pfp}
Let $p(n) = c n^{c-1}$ and let $\mathcal{M}$ be a $\twr{p(n)}{k}$-space-bounded DTM that decides a query over ordered LTS.
Let $Q$ be its state set and let $\Gamma$ be its tape alphabet.
Let $\Transsys$ be an LTS ordered by $<$ and such that its state set satisfies $|\States| \geq \max\{|Q|,|\Gamma|, c\}$. 
Let $w = \langle \Transsys\rangle$. Then 
\[
\Transsys \models \varphi_{\mathcal{M},w} \text{ iff } w \in L(\mathcal{M}).
\]
\end{lemma}
This follows from the previous lemmas. $\psi$ stipulates that either $X$ is empty, and a tuple is in its ``return value'' iff
it is in the encoding of the initial configuration, using $\varphi^w_{\text{init}}$, or it defers to $\varphi_{\text{trans}}$. The formula
$\varphi_{\mathcal{M}}$ then encodes the unique run of $\mathcal{M}$ on input $w$, by asking whether a tuple containing the accepting state 
is contained in the partial fixpoint of $\psi$. This is the case if and only if the machine halts in the accepting state, due to Lem.~\ref{lem:trans} 
and our observations on $\varphi^w_{\text{init}}$.

We omit the tedious, but standard argument that $\varphi_{\text{init}}^w$ can be rewritten into some $\varphi_{\text{init}}$ not depending on
$w$ that internalizes the translation from $\Transsys$ to $\langle \Transsys \rangle$.

\begin{theorem}
\hopfp{k+1} captures $k$-\ExpSpace over ordered LTS for $k \geq 2$.
\end{theorem}
One direction is by Thm.~\ref{thm:upper}, the other direction is by the previous Lem.~\ref{lem:pfp} plus the observation that LTS
that are smaller than in the requirements of the lemma can be enumerated in a constant-size formula.

\section{Conclusion}
\label{sec:concl}
We have shown that, over ordered structures, the queries expressible in $\hopfp{k+1}$ are exactly those decided
by a $\twr{p(n)}{k}$-space-bounded DTM, i.e., that $\hopfp{k+1}$ captures $k$-\ExpSpace over ordered structures for
$k \geq 0$, extending the same result by Vardi for $k=0$ \cite{DBLP:conf/stoc/Vardi82}. 

It should be noted that the requirement that the structures in question be ordered can be removed for $k \geq 1$,
as $\hopfp{2}$ and above possess sufficient expressive power to ``guess'' an order, cf.\ Fagin's Theorem \cite{Fagin74}.

Our result has applications in descriptive complexity. Otto's Theorem \cite{Otto99} characterizes
\emph{bisimulation-invariant} \PTime-queries as exactly those expressible in the polyadic modal $mu$-calculus. Contrary
to Immerman's and Vardi's characterization \cite{DBLP:journals/iandc/Immerman86, DBLP:conf/stoc/Vardi82} of PTIME, the crucial
requirement that the LTS be ordered is absent from this result, since an order can be recuperated in the bisimulation-invariant 
setting. However, the result makes use of the Immerman-Vardi Theorem.
We have extended this result to a characterization of bisimulation-invariant $k$-\ExpTime \cite{DBLP:journals/corr/abs-2209-10311}
using Freire and Martin's characterization of $k$-\ExpTime \cite{FREIRE201171}, i.e., their generalization of the Immerman-Vardi Theorem. 
The results of this paper open
up a similar characterization of the bisimulation-invariant exponential-space hierarchy, following from the second author's
Master's thesis \cite{KronenbergerMSc19}.

\bibliographystyle{eptcs}
\bibliography{biblio.bib}
\end{document}